\pdfoutput=1
\RequirePackage{ifpdf}
\ifpdf 
\documentclass[pdftex]{sigma}
\else
\documentclass{sigma}
\fi

\numberwithin{equation}{section}

\newcommand{\dd}{\mathrm{d}}
\newcommand{\Res}{\mathop{\,\rm Res\,}}

\usepackage{amsfonts,esint}

\newtheorem{Theorem}{Theorem}[section]
\newtheorem*{Theorem*}{Theorem}
\newtheorem{Corollary}[Theorem]{Corollary}
\newtheorem{Lemma}[Theorem]{Lemma}

 { \theoremstyle{definition}

\newtheorem{Remark}[Theorem]{Remark} }

\begin{document}


\newcommand{\arXivNumber}{2306.01501}

\renewcommand{\PaperNumber}{050}

\FirstPageHeading

\ShortArticleName{A Note on BKP for the Kontsevich Matrix Model with Arbitrary Potential}

\ArticleName{A Note on BKP for the Kontsevich Matrix Model\\ with Arbitrary Potential}

\Author{Ga\"etan BOROT~$^{\rm a}$ and Raimar WULKENHAAR~$^{\rm b}$}

\AuthorNameForHeading{G.~Borot and R.~Wulkenhaar}

\Address{$^{\rm a)}$~Institut f\"ur Mathematik und Institut f\"ur Physik, Humboldt-Universit\"at zu Berlin,\\
\hphantom{$^{\rm a)}$}~Unter den Linden~6, 10099 Berlin, Germany}
\EmailD{\href{mailto:gaetan.borot@hu-berlin.de}{gaetan.borot@hu-berlin.de}}

\Address{$^{\rm b)}$~Mathematisches Institut, Universit\"at M\"unster, Einsteinstr.\ 62, 48149 M\"unster, Germany}
\EmailD{\href{mailto:raimar@math.uni-muenster.de}{raimar@math.uni-muenster.de}}

\ArticleDates{Received January 03, 2024, in final form June 01, 2024; Published online June 11, 2024}

\Abstract{We exhibit the Kontsevich matrix model with arbitrary potential as a~BKP tau-function with respect to polynomial deformations of the potential. The result can be equivalently formulated in terms of Cartan--Pl\"ucker relations of certain averages of Schur $Q$-function. The extension of a Pfaffian integration identity of de Bruijn to singular kernels is instrumental in the derivation of the result.}

\Keywords{BKP hierarchy; matrix models; classical integrability}

\Classification{37K10; 37K20; 15A15}

\section{The formula}\label{Chap1}

\subsection{The Kontsevich matrix model with arbitrary potential}

Let $\mathcal{H}_N$ be the space of Hermitian $N \times N$ matrices
equipped with the Lebesgue measure
\[
\dd H= \prod_{i=1}^N \dd H_{ii}
 \prod_{1 \leq i < j \leq N} \dd\operatorname{Re}(H_{ij})\, \dd\operatorname{Im}(H_{ij}).
\]
Given a positive matrix $\Lambda = \operatorname{diag}(\lambda_1,\dots,\lambda_N)$,
we introduce the Gaussian probability measure on $\mathcal{H}_N$
\begin{equation}
\label{theGauss}
\dd \mathbb{P}_N(H) = \frac{\sqrt{\Delta(\boldsymbol{\lambda},\boldsymbol{\lambda})}}{2^{\frac{N}{2}} (2\pi)^{\frac{N^2}{2}}}\,\dd H {\rm e}^{-\frac{1}{2}{\rm Tr}(\Lambda H^2)}.
\end{equation}
We use the notations $\Delta(\boldsymbol{\lambda}) = \prod_{1 \leq i < j \leq N} (\lambda_j - \lambda_i)$ and $\Delta(\boldsymbol{\lambda},\boldsymbol{\mu}) = \prod_{1 \leq i,j \leq N} (\lambda_i + \mu_j)$. We denote $[n] = \{1,\dots,n\}$ and $\lambda_{\min} = \min\{\lambda_i\mid i \in [N]\}$.

Let $V_0$ be a continuous function on $\mathbb{R}$ such that the measure ${\rm e}^{- \frac{1}{2}\lambda_{\min} x^2 + V_0(x)}\,\dd x$ has finite moments on $\mathbb{R}$ (take for instance $V_0$ to be a polynomial of even degree with negative top coefficient). Then the measure $\dd \mathbb{P}_N(H) {\rm e}^{ \operatorname{Tr}V_0(H)}$ on $\mathcal{H}_N$ is finite. Let
\[
V_{\mathbf{t}}(x) = V_0(x) + \sum_{k \geq 0} t_{2k + 1} x^{2k + 1},
\]
where $\mathbf{t} = (t_{2k + 1})_{k \geq 1}$ are formal parameters. The partition function of the Kontsevich model with arbitrary potential is defined by
\begin{equation}
\label{Zdef} Z_N(\mathbf{t}) = \int_{\mathcal{H}_N} \dd\mathbb{P}_N(H) {\rm e}^{\operatorname{Tr}V_{\mathbf{t}}(H)}.
\end{equation}
This note aims at showing that $Z_N(\mathbf{t})$ is a
tau-function of the BKP hierarchy \cite{Date:1981qz}.

\subsection{Pfaffian formula}

First, we establish a Pfaffian formula for $Z_N(\mathbf{t})$. The proof proposed in Section~\ref{S2} consists in classical algebraic manipulations with matrix integrals and an analysis argument -- that one may find of independent interest, see Lemma~\ref{HUHH} and Remark~\ref{therem} -- to justify the extension of de~Bruijn's Pfaffian formula \cite{deBruijn1955} to singular kernels like the one appearing in \eqref{BZ}.
\begin{Theorem}
\label{th:1} For even $N$, we have
\begin{equation}
\label{BZ} Z_N(\mathbf{t}) = \frac{\sqrt{\Delta(\boldsymbol{\lambda},\boldsymbol{\lambda})}}{2^{\frac{N^2}{2}} (2\pi)^{\frac{N}{2}}\prod_{n = 1}^{N - 1} n!} \operatorname{Pf}_{0 \leq m,n \leq N - 1}(K_{N;m,n}(\mathbf{t})),
\end{equation}
where $\fint = \lim_{\epsilon \rightarrow 0} \int_{|x + y| \geq \epsilon}$ is the Cauchy principal value integral and
\begin{gather*}
 K_{N;m,n}(\mathbf{t}) = \fint_{\mathbb{R}^2} \frac{x - y}{x + y} F_{N;m}(x)F_{N;n}(y) {\rm e}^{V_{\mathbf{t}}(x) + V_{\mathbf{t}}(y)}\,\dd x \dd y, \\
F_{N;n}(x) = x^{2n} + \frac{(-2)^n n!}{\Delta(\boldsymbol{\lambda})}\det\bigl(\lambda_i^0 \big| \lambda_i^1 \big| \cdots \big| \lambda_i^{n - 1} \big| R_N\bigl(-\tfrac{1}{2}\lambda_ix^2\bigr) \big| \lambda_i^{n + 1} \big| \cdots \big| \lambda_i^{N - 1}\bigr), \\
R_N(\xi) = \frac{\xi^N}{(N - 1)!} \int_{0}^{1} \dd u\, (1 - u)^{N - 1} {\rm e}^{\xi u} = {\rm e}^{\xi}\biggl(1 - \frac{\Gamma(N;\xi)}{(N - 1)!}\biggr).
\end{gather*}
\end{Theorem}

Note that $R_N(\xi)$ is simply the $N$-th remainder of the Taylor series of ${\rm e}^{\xi}$. The expression given for $F_{N;n}(x)$ emphasises that it is $x^{2n} + O\bigl(x^{2N + 2}\bigr)$ as $x \rightarrow 0$, but we also have the equivalent expression
\[
F_{N;n}(x) = \frac{(-2)^n n!}{\Delta(\boldsymbol{\lambda})} \det\bigl(\lambda_i^0 \big| \lambda_i^1 \big| \cdots \big| \lambda_i^{n - 1} \big| {\rm e}^{-\frac{1}{2}\lambda_i x^2} \big| \lambda_i^{n + 1} \big| \cdots \big| \lambda_i^{N - 1}\bigr).
\]
For instance, the formula for $N = 2$ involves the two functions
\[
F_{2;0}(x) = \frac{{\rm e}^{-\frac{1}{2}\lambda_1 x^2} \lambda_2 - {\rm e}^{-\frac{1}{2}\lambda_2 x^2}\lambda_1}{\lambda_2 - \lambda_1}, \qquad F_{2;1}(x) = 2 \frac{{\rm e}^{-\frac{1}{2}\lambda_1 x^2} - {\rm e}^{-\frac{1}{2}\lambda_2 x^2}}{\lambda_2 - \lambda_1} .
\]

\subsection{BKP hierarchy}

If we ignored regularisation of the integrals, we would recognise in \eqref{BZ} the expression of a~BKP tau-function according to
\cite{Nimmo_1990}, see also \cite[Section~7.1.2.3]{harnad_balogh_2021}, where the functions $y_k$ of equation~(7.1.49) should be taken to
\[
y_k(\mathbf{t}) = \int_{\mathbb{R}} F_{N;k}(x) {\rm e}^{V_{\mathbf{t}}(x)}\,\dd x
\]
and depend on $\Lambda$. Yet, the presence of a regularisation
requires some care and the existing results in loc.\ cit.\ cannot be
applied as such (cf.\ Remark~\ref{RemarkBKP}). We propose in
Section~\ref{SproofBKP} an adaptation of the usual proof that works in
presence of Cauchy principal integrals. The obtained
Lemma~\ref{lemfinal} covers our case.

\begin{Corollary}
 \label{Comain} For fixed $\Lambda$, the formal power series
 $Z_N(\mathbf{t})$ is a BKP tau-function with respect to the times
 $\mathbf{t}$, i.e., it satisfies the Hirota bilinear equation of
 type B
\begin{equation}
\label{biline}
Z_N(\mathbf{t})Z_N\bigl(\tilde{\mathbf{t}}\bigr) = \mathop{{\rm Res}}_{z = 0} \frac{\dd z}{z} {\rm e}^{\sum_{k \geq 0} z^{2k + 1}(t_{2k + 1} - \tilde{t}_{2k + 1})} Z_N\bigl(\mathbf{t} - 2\bigl[z^{-1}\bigr]\bigr) Z_N\bigl(\tilde{\mathbf{t}} + 2\bigl[z^{-1}\bigr]\bigr),
\end{equation}
where $\bigl[z^{-1}\bigr] = \bigl(\frac{1}{z},\frac{1}{3z^3},\frac{1}{5z^5},\dots\bigr)$.
\end{Corollary}

Expanding the bilinear equation \eqref{biline} near $\mathbf{t} = \tilde{\mathbf{t}} = 0$ yields the BKP hierarchy of equations. In the present case, they give algebraic relations between the odd moments (or the cumulants) of the Kontsevich model with arbitrary potential, defined as
\begin{gather*}
M_{2\ell_1 +1 ,\dots,2\ell_n + 1} = \frac{1}{Z_N(\mathbf{t})} \frac{\partial}{\partial t_{2\ell_1 + 1}} \cdots \frac{\partial}{\partial t_{2\ell_n + 1}} Z_N(\mathbf{t}) \bigg|_{\mathbf{t} = 0} \\
\hphantom{M_{2\ell_1 +1 ,\dots,2\ell_n + 1}}{}
 = \frac{\int_{\mathcal{H}_N} \dd\mathbb{P}_N(H) {\rm e}^{{\rm Tr}(-\frac{1}{2}\Lambda H^2 + V_0(H))} \operatorname{Tr}H^{2\ell_1 + 1} \cdots \operatorname{Tr}H^{2\ell_n + 1}}{\int_{\mathcal{H}_N} \dd\mathbb{P}_N(H) {\rm e}^{{\rm Tr}(-\frac{1}{2}\Lambda H^2 + V_0(H))}}, \\
K_{2\ell_1 + 1,\dots,2\ell_n + 1} = \frac{\partial}{\partial t_{2\ell_1 + 1}} \cdots \frac{\partial}{\partial t_{2\ell_n + 1}} \ln Z(\mathbf{t})\bigg|_{\mathbf{t} = 0} \\
\hphantom{K_{2\ell_1 + 1,\dots,2\ell_n + 1}}{}
 = \sum_{\mathbf{I} = {\rm partitions}\,\,{\rm of}\,\,[n]} (-1)^{|\mathbf{I}| - 1}(|\mathbf{I}| - 1)! \prod_{I \in \mathbf{I}} M_{(\ell_i)_{i \in I}}.
\end{gather*}
Note that the hierarchy of equations \eqref{biline} does not depend on $N$, $\Lambda$ and $V_0$, though the particular solution $Z_N(\mathbf{t})$ does through the initial data $Z_N(\mathbf{0})$.

For instance, the first two BKP-equations are
\begin{gather}
0 = \bigl(D_1^6-5 D_1^3 D_3-5 D_3^2 +9 D_1 D_5 \bigr)(Z_N,Z_N)(\mathbf{0}), \nonumber \\
0 = \bigl(D_1^8 + 7D_1^5D_3 -35D_1^2D_3^2 - 21D_1^3D_5 - 42D_3D_5 + 90D_1D_7\bigr)(Z_N,Z_N)(\mathbf{0})\label{D1D1}
\end{gather}
in terms of the Hirota operators
$D_k(\tau,\tau)(\mathbf{t}) =\bigl(\partial_{t_{k}}-\partial_{\tilde{t}_k}\bigr)
\tau(\mathbf{t})\tau\bigl(\tilde{\mathbf{t}}\bigr)\big|_{\tilde{\mathbf{t}} = \mathbf{t}}$. For \emph{even} $V_0$, we have $M_{2\ell_1 + 1,\dots,2\ell_n + 1} = 0$ for $n$ odd, and \eqref{D1D1} results in
\begin{gather}
0 = M_{1^6} + 15 M_{1^4} M_{1,1} - 5 M_{3,1^3} - 15 M_{3,1}M_{1,1} - 5M_{3,3} + 9 M_{5,1}, \nonumber\\
0 = M_{1^8} + 28M_{1^6}M_{1,1} + 35(M_{1^4})^2 + 7M_{3,1^5} + 70M_{3,1^3}M_{1,1} + 35M_{3,1}M_{1^4} \nonumber\\
\hphantom{0 =}{}
- 35M_{3,3}M_{1,1} - 70(M_{3,1})^2 - 21 M_{5,1^3} - 63M_{5,1}M_{1,1} - 42M_{5,3} + 90 M_{7,1} .\label{Hirota-6}
\end{gather}
Or, equivalently in terms of cumulants,
\begin{gather*}
0 = K_{1^6} + 30K_{1^4}K_{1,1} + 60(K_{1,1})^3 - 5K_{3,1^3} - 5K_{3,3} - 30K_{3,1}K_{1,1} + 9K_{5,1}, \\
0 = K_{1^8} + 56 K_{1^6}K_{1,1} + 70(K_{1^4})^2 + 840 K_{1^4}(K_{1,1})^2 + 840(K_{1,1})^4 + 7K_{3,1^5} \\
\hphantom{0 =}{} + 70K_{3,1}K_{1^4} + 420 K_{3,1}(K_{1,1})^2 + 140K_{3,1^3}K_{1,1} - 35 K_{3,3}K_{1,1} - 70 (K_{3,1})^2 \\
\hphantom{0 =}{} - 21K_{5,1^3} -126 K_{5,1}K_{1,1} - 42 K_{5,3} + 90 K_{7,1}.
\end{gather*}
When $V_0$ is not even, many more terms contribute.

It is instructive to test these equation in the simplest case $V_0 = 0$. The moments can be found,\footnote{Note that in \cite{Mironov:2020tjf}, equation (45) follows from substituting $\Lambda \rightarrow \frac{1}{2}\Lambda$ in equation~(44). Equation~(45) is the one they use to compute moments, and the Gaussian probability measure it induces agrees with our $\mathbb{P}_N$ defined in~\eqref{theGauss}.} e.g., in~\cite{Mironov:2020tjf}, with
$p_k = \operatorname{Tr}\Lambda^{-k}$
\begin{alignat*}{3}
&M_{1,1} = p_1, && & \\
&M_{3,1} = 3 p_1^2, \qquad && M_{1^4} = 3 p_1^2, \\
&M_{1^6} = 15 p_1^3, \qquad && M_{3,1^3} = 6 p_3+9 p_1^3, \\
&M_{3,3} =3 p_3+12 p_1^3, \qquad && M_{5,1} = 5 p_3+10 p_1^3.
\end{alignat*}
They satisfy the first equation of \eqref{Hirota-6}, as expected.

\section{Proof of Theorem~\ref{th:1}}
\label{S2}
A Hermitian matrix $H$ decomposes as
\begin{gather*}
H = \sum_{a = 1}^N (\operatorname{Re} H_{a,a}) E_{a,a}+ \sum_{1 \leq a < b \leq N} \bigl(\sqrt{2} \operatorname{Re} H_{a,b}\bigr) \tfrac{1}{\sqrt{2}}(E_{a,b} + E_{b,a}) \\
\hphantom{H =}{}
 + \bigl(\sqrt{2} \operatorname{Im} H_{a,b}\bigr) \tfrac{{\rm i}}{\sqrt{2}}(E_{a,b} - E_{b,a}).
\end{gather*}
As $E_{a,a}$, $\frac{1}{\sqrt{2}}(E_{a,b} + E_{b,a})$ and $\frac{{\rm i}}{\sqrt{2}}(E_{a,b} - {\rm E}_{b,a})$ have unit norm for the standard Euclidean metric on ${\rm Mat}_{N}(\mathbb{C}) \cong \mathbb{R}^{2N^2}$, the volume form on $\mathcal{H}_N$ induced by the Euclidean volume form on ${\rm Mat}_N(\mathbb{C}) \cong \mathbb{R}^{2N^2}$ is $2^{\frac{N(N - 1)}{2}}\,\dd H$. Denote $\mathcal{U}_N$ the unitary group and $\dd\nu$ its volume form induced by the Euclidean volume form in ${\rm Mat}_N(\mathbb{C})$. The corresponding volume is
 \[
 {\rm Vol}(\mathcal{U}_N) = \frac{(2\pi)^{\frac{N(N + 1)}{2}}}{\prod_{n = 1}^{N - 1} n!}.
\]
We also recall the Harish-Chandra--Itzykson--Zuber formula \cite{Itzykson:1979fi}
\[
\frac{1}{\prod_{n = 1}^{N - 1} n!} \int_{\mathcal{U}_N} \frac{\dd \nu(U)}{{\rm Vol}(\mathcal{U}_N)} {\rm e}^{{\rm Tr}(AUBU^{\dagger})} = \frac{\det\bigl({\rm e}^{a_ib_j}\bigr)}{\Delta(\boldsymbol{a})\Delta(\boldsymbol{b})},
\]
 where $A = \operatorname{diag}(a_1,\dots,a_N)$ and $B = \operatorname{diag}(b_1,\dots,b_N)$.

Diagonalising the matrix $H = UXU^{\dagger}$ with $X = \operatorname{diag}(x_1,\dots,x_N)$ and $U \in \mathcal{U}_N$ defined up to action of $\mathfrak{S}_N \times \mathcal{U}_1^{N}$ brings the partition function \eqref{Zdef} in the form
\begin{align}
Z_N(\mathbf{t}) & = \frac{\sqrt{\Delta(\boldsymbol{\lambda},\boldsymbol{\lambda})}}{2^{\frac{N}{2}} (2\pi)^{\frac{N^2}{2}}} \frac{1}{N! (2\pi)^N 2^{\frac{N(N - 1)}{2}}} \nonumber\\
&\quad{}\times \int_{\mathbb{R}^N} \biggl(\int_{\mathcal{U}_N} \dd \nu(U) {\rm e}^{-\frac{1}{2}{\rm Tr}(\Lambda U X^2 U^{\dagger})}\biggr) (\Delta(\boldsymbol{x}))^2 \prod_{i = 1}^N {\rm e}^{V_{\mathbf{t}}(x_i)}\dd x_i \nonumber \\
& = \frac{\sqrt{\Delta(\boldsymbol{\lambda},\boldsymbol{\lambda})}}{2^{\frac{N^2}{2}} (2\pi)^{\frac{N}{2}} N! \Delta\bigl(-\frac{\boldsymbol{\lambda}}{2}\bigr)} \int_{\mathbb{R}^N} \frac{(\Delta(\boldsymbol{x}))^2}{\Delta\bigl(\boldsymbol{x}^2\bigr)} \mathop{\det}_{1 \leq i,j \leq N} \bigl({\rm e}^{-\frac{1}{2}\lambda_ix_j^2}\bigr) \prod_{i = 1}^N {\rm e}^{V_{\mathbf{t}}(x_i)} \dd x_i. \label{Ztint}
\end{align}
Here we could use Fubini because the integrand in the first line of \eqref{Ztint} is real positive, and in fact integrable due to the assumptions on $V_0$ and $\Lambda$. We observe that \smash{$\Delta\bigl(-\frac{\boldsymbol{\lambda}}{2}\bigr) = (-2)^{-\frac{N(N - 1)}{2}}\Delta(\boldsymbol{\lambda})$} and recall Schur's Pfaffian identity \cite{Schur1911}, for $N$ even
\[
\frac{(\Delta(\boldsymbol{x}))^2}{\Delta\bigl(\boldsymbol{x}^2\bigr)} = \prod_{1 \leq i < j \leq N} \frac{x_j - x_i}{x_j + x_i} = \operatorname{Pf}_{1 \leq i,j \leq N}\biggl(\frac{x_j - x_i}{x_j + x_i}\biggr).
\]
So, up to a prefactor, $Z_N(\mathbf{t})$ is an integral of the form
\begin{equation}
\label{Sint}
\int_{\mathbb{R}^N} \operatorname{Pf}_{1 \leq i,j \leq N} (S(x_i,x_j)) \mathop{{\rm det}}_{\substack{0 \leq m \leq N - 1 \\ 1 \leq j \leq N}} \bigl(f_m\bigl(x_j^2\bigr)\bigr) \prod_{i = 1}^N \rho(x_i)\,\dd x_i.
\end{equation}
De Bruijn's identity \cite{deBruijn1955} would allow rewriting \eqref{Sint} as
\begin{equation}
\label{dJ}
N! \operatorname{Pf}_{0 \leq m,n \leq N - 1}\biggl(\int_{\mathbb{R}^2} S(x,y) f_m\bigl(x^2\bigr)f_n\bigl(y^2\bigr)\rho(x)\rho(y)\,\dd x \dd y\biggr),
\end{equation}
but the proof in loc.\ cit.\ is solely based on algebraic manipulations, valid when $(f_n)_{n = 0}^{N - 1}$ is a~sequence of measurable functions on $\mathbb{R}_{\geq 0}$ and $S(x,y) = -S(y,x)$ is a measurable function on $\mathbb{R}^2$ such that $\int_{\mathbb{R}^2} \big|S(x,y)f_m\bigl(x^2\bigr)f_n\bigl(y^2\bigr)\big|\rho(x)\rho(y)\dd x\dd y < +\infty$. The choice of $S(x,y) = \frac{x - y}{x + y}$ in general violates this integrability assumption due to the presence of the simple pole on the anti-diagonal combined with the non-compactness of $\mathbb{R}^2$. Nevertheless, we show that the conclusion~\eqref{dJ} remains valid provided the integral in the Pfaffian is understood as a Cauchy principal value, under a Schwartz-type condition.

\begin{Lemma}
\label{HUHH}Let $\rho > 0$ be a measurable function on $\mathbb{R}$ and $\mathbb(f_n)_{n = 0}^{N - 1}$ be a sequence of $\mathcal{C}^{N - 1}$-functions on $\mathbb{R}_{\geq 0}$ such that $f^{(\ell)}_m$ is bounded by a polynomial for any $m,\ell \in \{0,\dots,N - 1\}$. Let $S(x,y) = \frac{\tilde{S}(x,y)}{x + y}$ where $\tilde{S}$ is a measurable function on $\mathbb{R}^2$ such that
\[
\forall k,l \in \mathbb{Z}_{\geq 0}\qquad \int_{\mathbb{R}^2} \bigl|\tilde{S}(x,y)x^ky^l\bigr|\rho(x)\rho(y)\,\dd x \dd y < +\infty.
\]
Then, for $N$ even
\begin{gather*}
 \int_{\mathbb{R}^N} \operatorname{Pf}_{1 \leq i,j \leq N} (S(x_i,x_j)) \mathop{\det}_{\substack{0 \leq m \leq N - 1 \\
 1 \leq j \leq N}} \bigl(f_m\bigl(x_j^2\bigr)\bigr) \prod_{i = 1}^n \rho(x_i)\dd x_i \\
\qquad\quad{} = N!\operatorname{Pf}_{0 \leq m,n \leq N - 1}\biggl(\fint_{\mathbb{R}^2} S(x,y) f_m\bigl(x^2\bigr) f_n\bigl(y^2\bigr)\rho(x)\rho(y)\,\dd x \dd y\biggr),
\end{gather*}
where $\fint = \lim_{\epsilon \rightarrow 0} \int_{|x + y| \geq \epsilon}$ and the integrand in the left-hand side is integrable.
\end{Lemma}
\begin{proof} Take $\epsilon > 0$ and set $S_{\epsilon}(x,y) = S(x,y) \cdot \mathbf{1}_{|x + y| \geq \epsilon}$. In this situation, we can use de Bruijn's formula and write
\begin{gather}
\nonumber \int_{\mathbb{R}^N} \operatorname{Pf}_{1 \leq i,j \leq N} (S_{\epsilon}(x_i,x_j)) \mathop{\det}_{\substack{0 \leq m \leq N - 1 \\ 1 \leq j \leq N}} \bigl(f_m\bigl(x_j^2\bigr)\bigr) \prod_{i = 1}^N \rho(x_i)\dd x_i \\
\qquad {} = N!\operatorname{Pf}_{0 \leq m,n \leq N - 1} \biggl(\int_{\mathbb{R}^2} S_{\epsilon}(x,y)f_m\bigl(x^2\bigr)f_n\bigl(y^2\bigr)\rho(x)\rho(y)\,\dd x \dd y\biggr).\label{LHs}
\end{gather}
The right-hand side tends to
\[
N! \operatorname{Pf}_{0 \leq m,n \leq N - 1} \biggl(\fint_{\mathbb{R}^2} S(x,y) f_m\bigl(x^2\bigr)f_n\bigl(y^2\bigr)\rho(x)\rho(y)\dd x \dd y\biggr)
\]
 when ${\epsilon \rightarrow 0}$. Call $I_{\epsilon}(\mathbf{x})$ the integrand in the left-hand side of formula~\eqref{LHs}. We clearly have $\lim_{\epsilon \rightarrow 0} I_{\epsilon}(\mathbf{x}) = I_0(\mathbf{x})$ for $\mathbf{x}$ almost everywhere in $\mathbb{R}^N$. Provided we can find for $I_{\epsilon}(\mathbf{x})$ a uniform in $\epsilon$ and integrable on~$\mathbb{R}^2$ upper bound, the lemma follows from dominated convergence.

To find such a bound, we introduce the matrix $W(\boldsymbol{\xi})$ with entries $\xi_j^{n}$ at row index $n \in \{0,\dots,N - 1\}$ and column index $j \in [N]$, which satisfies $\Delta(\boldsymbol{\xi}) = \det W(\boldsymbol{\xi})$. Its inverse matrix is
\[
\bigl(W(\boldsymbol{\xi})^{-1}\bigr)_{i,n} = \frac{(-1)^{N - n - 1}e_{N - n - 1}\bigl(\boldsymbol{\xi}_{[i]}\bigr)}{\prod_{j \neq i} (\xi_i - \xi_j)},
\]
where $e_k$ is the $k$-th elementary symmetric polynomial and $\boldsymbol{\xi}_{[i]} = \bigl(\xi_1,\dots,\widehat{\xi_i},\dots,\xi_N\bigr)$. Then
\begin{gather}
\nonumber \mathop{\det}_{\substack{0 \leq m \leq N - 1 \\ 1 \leq j \leq N }} \bigl(f_m\bigl(x_j^2\bigr)\bigr) = \Delta\bigl(\boldsymbol{x}^2\bigr) \det_{\substack{1 \leq i \leq N \\ 0 \leq n \leq N - 1}} \bigl(f_n\bigl(x_i^2\bigr)\bigr) \cdot \det \bigl(W\bigl(\boldsymbol{x}^2\bigr)^{-1}\bigr) \\
\hphantom{\mathop{\det}_{\substack{0 \leq m \leq N - 1 \\ 1 \leq j \leq N }} \bigl(f_m\bigl(x_j^2\bigr)\bigr)}{} = \Delta\bigl(\boldsymbol{x}^2\bigr) \det_{0 \leq m,n \leq N - 1} \Biggl(\sum_{i = 1}^{N} \frac{(-1)^{N - m - 1} f_n\bigl(x_i^2\bigr) e_{N - m - 1}\bigl(\boldsymbol{x}^2_{[i]}\bigr)}{\prod_{j \neq i} \bigl(x_i^2 - x_j^2\bigr)}\Biggr).\label{26}
\end{gather}
Up to a sign that we can take out of the determinant, the $(m,n)$-entry inside the determinant is
\begin{align*}
 \bigl[u^{N - m - 1}\bigr] \sum_{i = 1}^N f_n\bigl(x_i^2\bigr) \prod_{j \neq i} \frac{1 + ux_j^2}{x_i^2 - x_j^2} & = \bigl[u^{N - m - 1}\bigr] \prod_{i = 1}^N \bigl(1 + ux_i^2\bigr)\!\Biggl( \sum_{i = 1}^N \frac{f_n\bigl(x_i^2\bigr)}{1 + ux_i^2} \frac{1}{\prod_{j \neq i} (x_i^2 - x_j^2)}\Biggr) \\
& = \sum_{k = 0}^{N - m - 1} e_{N - m - 1 - k}\bigl(\boldsymbol{x}^2\bigr)\!\Biggl( \sum_{i = 1}^{N} \frac{(-1)^k x_i^{2k} f_{n}\bigl(x_i^2\bigr)}{\prod_{j \neq i} \bigl(x_i^2 - x_j^2\bigr)}\Biggr).
\end{align*}
In the first two steps, $[u^m]$ acting on the formal power series of $u$ to its right meant extracting the coefficient of $u^m$. Up to the use of squared variables, we recognise the divided difference
\[
g[\xi_1,\dots,\xi_N] := \sum_{i = 1}^N \frac{g(\xi_i)}{\prod_{j \neq i} (\xi_i - \xi_j)}.
\]
When $g$ is $\mathcal{C}^{N - 1}$, it can be written (see, e.g., \cite[Theorem 2, p.~250]{AnaN}) as an integral over the $(N - 1)$-dimensional simplex $\Delta_{N - 1} = \{p \in [0,1]^N \mid p_1 + \cdots + p_N = 1\}$, equipped with the volume form $\dd \sigma(\boldsymbol{p}) = \dd p_1 \cdots \dd p_{N - 1}$:
\[
g[\xi_1,\dots,\xi_N] = \int_{\Delta_{N - 1}} g^{(N - 1)}(p_1\xi_1 + \cdots + p_N\xi_N)\,\dd \sigma(\boldsymbol{p}).
\]
We use this for $g_{k,n}(\xi) = (-1)^{k}\xi^k f_n(\xi)$. Inserting the integral representation in \eqref{26} yields
\begin{gather*}
|I_{\epsilon}(\mathbf{x})| = \bigl|\Delta\bigl(\boldsymbol{x}^2\bigr)\bigr| \operatorname{Pf}_{1 \leq i,j \leq N}(S_{\epsilon}(x_i,x_j) ) \\
 \times \Biggl| \mathop{{\rm det}}_{0 \leq m,n \leq N - 1}\!\Biggl( \sum_{k = 0}^{N - 1 - m}\! e_{N - 1 - m - k}\bigl(\boldsymbol{x}^2\bigr) \int_{\Delta_{N - 1}} g_{k,n}^{(N - 1)}\bigl(p_1x_1^2 + \cdots + p_N x_N^2\bigr) \dd \sigma(\boldsymbol{p})\Biggr)\Biggr| \prod_{i = 1}^N \rho(x_i).
\end{gather*}
Since $\bigl|\Delta\bigl(\boldsymbol{x}^2\bigr)\bigr|$ cancels the denominators in $S_{\epsilon}$, the first line of the right-hand side admits an upper bound
by sum of terms, each of which is a polynomial in $\boldsymbol{x}$ multiplied by \smash{$\prod_{\{i,j\} \in \mathcal{P}} \bigl|\tilde{S}(x_i,x_j)\bigr|$}, where $\mathcal{P}$ is a partition of $[N]$ into pairs. In the second line, we first expand the determinant inside the absolute value and use the triangular inequality to get an upper bound by a sum of finitely many positive terms, each of which involves an $N$-fold product of simplex integrals of functions with at most polynomial growth, since the derivatives $f_{n}^{(\ell)}$ \big(and thus $g_{k,n}^{(N - 1)}$\big) have at most polynomial growth. Therefore, they result in a polynomial upper bound in the variable $\boldsymbol{x}$. We are thus left with an upper bound by a sum of finitely many terms of the form
\[
\prod_{\{i,j\} \in \mathcal{P}} \bigl|\tilde{S}(x_i,x_j)\bigr| \prod_{i = 1}^{N} x_i^{q_i}\rho(x_i) = \prod_{\{i,j\} \in \mathcal{P}} \bigl|\tilde{S}(x_i,x_j)\bigr|x_i^{q_i}x_j^{q_j} \rho(x_i)\rho(x_j)
\]
for various $N$-tuples of integers $\boldsymbol{q}$ and pair partitions $\mathcal{P}$ of $[N]$. Integrating each term of this form over $\mathbb{R}^{N}$ factorizes into a product of $\frac{N}{2}$ two-dimensional integrals, each of them being finite by assumption. This provides the domination assumption to conclude $\smash{\lim_{\epsilon \rightarrow 0} \int_{\mathbb{R}^N} I_{\epsilon}(\boldsymbol{x}) \prod_{i = 1}^N \dd x_i} = \smash{\int_{\mathbb{R}^N} I_0(\boldsymbol{x}) \prod_{i = 1}^N \dd x_i}$ as desired.
\end{proof}

\begin{Remark}
\label{therem} The proof can easily be adapted to obtain an analogous statement for kernels of~the form \smash{$S(x,y) = \frac{\tilde{S}(x,y)}{x - y}$}, in which case one can use $f_n(x)$ instead of $f_n\bigl(x^2\bigr)$.
\end{Remark}

The assumptions of Lemma~\ref{HUHH} are fulfilled for
\[
S(x,y) = \frac{x - y}{x + y},\qquad \rho(x) = {\rm e}^{-\frac{1}{2}\lambda_{\min} x^2 + V_{\mathbf{t}}(x)},\qquad f_m(\xi) = {\rm e}^{-\frac{1}{2}(\lambda_{m+1} - \lambda_{\min})\xi},
\]
where we stress that $\mathbf{t}$ are formal parameters. Therefore, coming back to \eqref{Ztint} and tracking the $N$-dependent prefactors, we arrive to the identity of formal power series in the variables $\mathbf{t}$:
\begin{gather*}
Z_N(\mathbf{t}) = \frac{(-1)^{\frac{N(N - 1)}{2}} \sqrt{\Delta(\boldsymbol{\lambda},\boldsymbol{\lambda})}}{2^N \pi^{\frac{N}{2}} \Delta(\boldsymbol{\lambda})} \operatorname{Pf}_{1 \leq m,n \leq N} \bigl(L_{m,n}\bigr), \\
L_{m,n} = \biggl(\fint_{\mathbb{R}^2} \frac{x - y}{x + y} {\rm e}^{-\frac{1}{2}\lambda_m x^2 - \frac{1}{2}\lambda_n y^2 + V_{\mathbf{t}}(x) + V_{\mathbf{t}}(y)} \,\dd x \dd y\biggr).
\end{gather*}

We would like to rewrite this formula by absorbing the denominator $\Delta(\boldsymbol{\lambda})$ in the Pfaffian. Recall the transposed Vandermonde matrix $W(\boldsymbol{\lambda})^{\mathsf T}$, whose entries are $W(\boldsymbol{\lambda})_{i,n}^{\mathsf T} = \lambda_i^{n}$ indexed by $i \in [N]$ and $n \in \{0,\dots,N - 1\}$. We have
\[
\frac{{\rm Pf}(L)}{\Delta(\boldsymbol{\lambda})} = \frac{{\rm Pf}(L)}{\det W(\boldsymbol{\lambda})^{\mathsf T}} = {\rm Pf}\bigl(\bigl(W(\boldsymbol{\lambda})^{\mathsf T}\bigr)^{-1} LW(\boldsymbol{\lambda})^{-1}\bigr),
\]
and by Cramer's formula for the inverse
\[
\bigl(\bigl(W(\boldsymbol{\lambda})^{\mathsf T}\bigr)^{-1} L W(\boldsymbol{\lambda})^{\mathsf T}\bigr)_{m,n} = v_mv_n\fint_{\mathbb{R}^2} \frac{x - y}{x + y} F_{N;m}(x)F_{N;n}(y) {\rm e}^{V_{\mathbf{t}}(x) + V_{\mathbf{t}}(y)}\,\dd x\dd y,
\]
where $v_n$ are non-zero constants to be chosen later, rows and columns are indexed by $m,n \in \{0,\dots,N - 1\}$, and we introduced
\begin{align*}
F_{N;m}(x) &{}= \frac{1}{v_m}\sum_{i = 1}^{N} \bigl(W(\boldsymbol{\lambda})^{\mathsf T}\bigr)^{-1}_{i,m} {\rm e}^{-\frac{1}{2}\lambda_i x^2} \\
&{}= \frac{\det\bigl(\lambda_i^0 \big| \lambda_i^1  \big| \cdots \big| \lambda_i^{m - 1} \big| {\rm e}^{-\frac{1}{2}\lambda_i x^2} \big| \lambda_i^{m + 1} \big| \cdots \big| \lambda_i^{N - 1}\bigr)}{v_m \Delta(\boldsymbol{\lambda})}.
\end{align*}
With Taylor formula in integral form at order $N$ near $0$, we can write
\begin{gather*}
{\rm e}^{-\frac{1}{2}\lambda_ix^2} = P_{N - 1}\bigl(-\tfrac{1}{2}\lambda_ix^2\bigr) + \frac{\bigl(-\frac{1}{2}\lambda_i x^2\bigr)^m}{m!} + R_{N}\bigl(-\tfrac{1}{2}\lambda_i x^2\bigr), \\
R_N(\xi) = \frac{\xi^{N}}{(N - 1)!} \int_{0}^{1}(1-u)^{N - 1} {\rm e}^{\xi u}\, \dd u
\end{gather*}
for some polynomial $P_{N - 1}$ of degree at most $N - 1$ and without its term of degree $m$ (which we wrote separately). The contribution of $P_{N - 1}$ disappears as it is a linear combination of the other columns, while the contribution of the degree $m$ term simply retrieves the Vandermonde determinant. Hence,
\[
F_{N;m}(x) = \frac{(-1)^mx^{2m}}{2^m m! v_m} + \frac{\det\bigl(\lambda_i^0  \big| \lambda_i^1  \big| \cdots \big| \lambda_i^{m - 1} \big| R_N\bigl(-\tfrac{1}{2}\lambda_ix^2\bigr) \big| \lambda_i^{m + 1} \big| \cdots \big| \lambda_i^{N - 1}\bigr)}{v_m \Delta(\boldsymbol{\lambda})}.
\]
We now choose $v_m = \frac{(-1)^m}{2^m m!}$ to get $F_{N;m}(x) = x^{2m} + O\bigl(x^{2N}\bigr)$ when $x \rightarrow 0$. Introducing the matrix
\[
K_{N;m,n}(\mathbf{t}) = \fint_{\mathbb{R}^2} \frac{x - y}{x + y}F_{N;m}(x)F_{N;n}(y) {\rm e}^{V_{\mathbf{t}}(x) + V_{\mathbf{t}}(y)}\, \dd x \dd y,
\]
we arrive to
\begin{equation*}
\begin{split}
Z_N(\mathbf{t}) & = \frac{(-1)^{\frac{N(N - 1)}{2}} \sqrt{\Delta(\boldsymbol{\lambda},\boldsymbol{\lambda})} \prod_{n = 0}^{N - 1} v_n}{2^N \pi^{\frac{N}{2}}}   \operatorname{Pf}_{0 \leq m,n \leq N - 1} K_{N;m,n}(\mathbf{t}) \\
& = \frac{\sqrt{\Delta(\boldsymbol{\lambda},\boldsymbol{\lambda})}}{2^{\frac{N^2}{2}} (2\pi)^{\frac{N}{2}} \prod_{n = 1}^{N - 1} n!}  \operatorname{Pf}_{0 \leq m,n \leq N - 1} K_{N;m,n}(\mathbf{t}).
\end{split}
\end{equation*}

 \section{Proof of Corollary~\ref{Comain}}
\label{SproofBKP}

\subsection{Preliminaries}
It is well-known that the BKP integrable hierarchy can be formulated
in terms of
neutral fermions $(\phi_j)_{j \in \mathbb{Z}}$ which satisfy
the anti-commutation relations
$ \{\phi_j,\phi_k\} =(-1)^j \delta_{j+k,0}$,
in particular $(\phi_0)^2=\tfrac{1}{2}$. There is a highest-weight representation of the algebra of neutral fermions, with highest-weight vector (or vacuum) $|0\rangle$
which satisfies
$\phi_{-j} |0\rangle =0 =\langle 0| \phi_{j}$ for $j>0$ and
$\langle 0|\phi_0|0\rangle=0$. The pair expectation values are
\begin{equation}
\label{pair-expectation}
 \langle 0| \phi_j \phi_k |0\rangle =
 \begin{cases}
 (-1)^k\delta_{j,-k}& \text{if}\ k>0,\\
\frac{1}{2}\delta_{j,0}& \text{if}\ k=0,\\
 0& \text{if}\ k<0.
 \end{cases}
 \end{equation}
 We introduce the generating series\footnote{Up to Lemma~\ref{lemmaRes},
 all statements hold for $x,x_i,x_i'\in \mathbb{C}$. The restriction to
 real variables is motivated by the intended integration.}
$ \phi(x): = \sum_{j \in \mathbb {Z}} x^j \phi_j $
for $x \in \mathbb{R}$. Vacuum expectations of
products of $\phi(x_i)$ are understood in a radial ordering. If
all $|x_i|$ are pairwise distinct, then
\begin{gather}
 \langle 0|\phi(x_1)\cdots \phi(x_{N})|0\rangle
:=(-1)^{\mathrm{sign}(\pi)}
\langle 0|\phi\bigl(x_{\pi(1)}\bigr)\cdots \phi\bigl(x_{\pi(n)}\bigr)|0\rangle\label{radialordering}
\end{gather}
if $\bigl|x_{\pi(1)}\bigr|>\bigl|x_{\pi(2)}\bigr|>\dots >\bigl|x_{\pi(N)}\bigr|$.

From \eqref{pair-expectation}, one finds
\begin{equation*}
 \langle 0|\phi(x_1)\phi(x_2)|0\rangle
 =\frac{1}{2} \frac{x_1-x_2}{x_1+x_2},
\end{equation*}
understood as convergent power series in
$\frac{x_1}{x_2}$ for $|x_1|<|x_2|$ and as convergent power series in
$\frac{x_2}{x_1}$ for $|x_2|<|x_1|$. The following is known as Wick's theorem:
For pairwise different $|x_i|$, one has for $N$ even
\begin{align}
\nonumber
 \langle 0|\phi(x_1)\phi(x_2)\cdots
 \phi(x_{N})|0\rangle
 & =\operatorname{Pf}_{1\leq k,l \leq N}
 (\langle 0|\phi(x_k)\phi(x_l)|0\rangle) \\
 &= \frac{1}{2^{\frac{N}{2}}}
 \operatorname{Pf}_{1\leq k,l \leq N}
 \biggl(\frac{x_k-x_l}{x_k+x_l}\biggr)
= \frac{1}{2^{\frac{N}{2}}}\prod_{1\leq k<l\leq N}
\frac{x_k-x_l}{x_k+x_l},\label{Wick}
\end{align}
and for $N$ odd $\langle 0|\phi(x_1)\phi(x_2)\cdots
 \phi(x_{N})|0\rangle=0$.

Next, consider the source operators
\begin{equation*}
\forall m \in \mathbb{Z}_{\geq 0}, \qquad J_m = \frac{1}{2} \sum_{j\in\mathbb{Z}}(-1)^j(\phi_{-j-m}\phi_{j}-
 \langle 0 |\phi_{-j-m}\phi_{j}|0\rangle).
\end{equation*}
One checks that all even $J_{2m}$ vanish identically, and that the $(J_{2m+1})_{m \geq 0}$ commute with each other. This gives rise to an infinite
family of commuting BKP flows
\[
\gamma(\mathbf{t}) := {\rm e}^{\sum_{m=0}^\infty J_{2 m+1} t_{2m+1} },\qquad
\mathbf{t}=(t_1,t_3,t_5,\dots).
\]
These satisfy $\gamma(\mathbf{t}) \gamma\bigl(\tilde{\mathbf{t}}\bigr)
= \gamma\bigl(\mathbf{t}+\tilde{\mathbf{t}}\bigr)$ and $\gamma(\mathbf{t})|0\rangle
=|0\rangle$. One finds $[J_{2m+1},\phi_j]=\phi_{j-(2m+1)}$ which leads to
$[J_{2m+1},\phi(x)]=x^{2m+1} \phi(x)$ and
\begin{equation}
\label{gammatphiz}
 \gamma(\mathbf{t}) \phi(x)=
 {\rm e}^{\sum_{m \geq 0} x^{2 m+1} t_{2m+1} } \phi(x)\gamma(\mathbf{t}).
\end{equation}

\subsection{A residue formula}

Let $N\in \mathbb{Z}_{\geq 0}$ be even. For pairwise distinct $|x_1|,\dots,|x_N|$
and a sequence of formal variables $\mathbf{t}=(t_1,t_3,\dots)$, we consider
\[
 \tau(\mathbf{t};x_1,\dots,x_N):= \langle 0 | \gamma(\mathbf{t})
 \phi(x_1)\cdots \phi(x_N)|0\rangle.
\]
Let $\bigl[z^{-1}\bigr]:=\bigl(\frac{1}{z},\frac{1}{3z^3},
\frac{1}{5 z^5},\dots\bigr)$.
 The group law implies
 \begin{equation}
 \tau\bigl(\mathbf{t}-2\bigl[z^{-1}\bigr];x_1,\dots,x_N\bigr) =\langle 0|
 \gamma\bigl(-2\bigl[z^{-1}\bigr]\bigr)
 \gamma(\mathbf{t}) \phi(x_1)\cdots \phi(x_N) |0\rangle.
\label{tau-shift1}
\end{equation}
When commuting $\gamma(-2[z^{-1}])$ to the right just before $|0\rangle$,
we see that the exponentials produced by \eqref{gammatphiz}
are well defined for $|z|>\max_i |x_i|$. We use
\cite[Lemma 7.3.9]{harnad_balogh_2021}
\[
 \langle 0|\gamma\bigl(-2\bigl[z^{-1}\bigr]\bigr)=2
 \langle 0|\phi_0 \phi(z)
\]
and
$\langle 0|\phi_0 a|0 \rangle=
\langle 0|a\phi_0 |0 \rangle$ for any element $a$ of the Clifford algebra
generated by $\phi_j$ (see, e.g., \cite[Exercise 7.5]{harnad_balogh_2021})
as well as \eqref{gammatphiz} to turns \eqref{tau-shift1} into
\begin{equation}
{\rm e}^{\sum_{m \geq 0} z^{2m+1}t_{2m+1}} \tau\bigl(\mathbf{t}-2\bigl[z^{-1}\bigr];x_1,\dots,x_N\bigr)
= 2 \langle 0|\gamma(\mathbf{t}) \phi(z)
\phi(x_1)\cdots \phi(x_N)\phi_0|0\rangle .
 \label{tau-shift2}
\end{equation}
Recall that $\tau\bigl(\mathbf{t}-2\bigl[z^{-1}\bigr];x_1,\dots,x_N\bigr)$ is a convergent
power series in $z^{-1}$ if $|z|>\max_i |x_i|$. At the very end the
$t_{2m+1}$ will be formal parameters. But at this intermediate point we
choose $|t_{2m+1}|<\tilde{R}^{-2m-1}$ for some large
$\tilde{R}> R\geq \max_i |x_i|$. Then \eqref{tau-shift2} is analytic
in $z$ in the domain $R<|z|<\tilde{R}$. We multiply by a second copy
$2 \langle 0|\gamma(\tilde{\mathbf{t}}) \phi(-z) \phi(x'_1)\cdots
\phi(x'_N)\phi_0|0\rangle$ with analogous parameter range. The
product is analytic in $z$ in the domain $\max_i \{|x_i|,|x_i'|\} \leq R <|z|<\tilde{R}$
and has there a convergent Laurent series expansion in $z$. We
multiply by $z^{-1}$, take the contour integral around a positively oriented circle $C$ in this annular domain, and make the radial
ordering \eqref{radialordering} explicit:
\begin{align*}
A:={}& \frac{1}{2{\rm i}\pi} \oint_{C} \biggl(\frac{\dd z}{z}
{\rm e}^{\sum_{m \geq 0} z^{2m+1}(t_{2m+1}-\tilde{t}_{2m+1})}
\tau\bigl(\mathbf{t}-2\bigl[z^{-1}\bigr];x_1,\dots,x_N\bigr)
\\ &
\times \tau\bigl(\tilde{\mathbf{t}}+2\bigl[z^{-1}\bigr];x_1',\dots,x_N'\bigr)\biggr)
\\
={}&  4 \oint_{C} \frac{\dd z}{z}
\langle 0|\gamma(\mathbf{t}) \phi(z) \phi(x_1)\cdots \phi(x_N) \phi_0|0\rangle
\langle 0|\gamma\bigl(\tilde{\mathbf{t}}\bigr) \phi(-z) \phi(x_1')\cdots \phi(x_N')
\phi_0|0\rangle.
\end{align*}
This would be evaluated to
\begin{gather}
A = 4 (-1)^{\mathrm{sign}(\pi)+\mathrm{sign}(\pi')}\label{pre-Hirotabis}
\\ \hphantom{A =}{}
 \times\sum_{j\in \mathbb{Z}}(-1)^j
\langle 0|\gamma(\mathbf{t}) \phi_j \phi\bigl(x_{\pi(1)}\bigr)\cdots \phi\bigl(x_{\pi(N)}\bigr)
\phi_0|0\rangle
\langle 0|\gamma\bigl(\tilde{\mathbf{t}}\bigr) \phi_{-j} \phi\bigl(x_{\pi'(1)}'\bigr)\cdots \phi\bigl(x_{\pi'(N)}'\bigr)
\phi_0|0\rangle\nonumber
\end{gather}
if the sum over $j$ converged absolutely, where we have introduced the permutations $\pi$, $\pi'$ such that
\[
R>\bigl|x_{\pi(1)}\bigr|>\cdots >\bigl|x_{\pi(N)}\bigr| \qquad \text{and} \qquad R>\bigl|x'_{\pi'(1)}\bigr|>\cdots >\bigl|x'_{\pi'(N)}\bigr|.
\]
The next lemma justifies the convergence under mild additional conditions on $x$'s and $x'$'s.
\begin{Lemma}\label{lemmaRes}
Let $\tilde{R} >\! R >\! 0$, assume $|t_{2m+1}|,\bigl|\tilde{t}_{2m+1}\bigr|\!<\! \tilde{R}^{-2m-1}$, that $|x_1|,\dots,|x_N|,|x_1'|,\dots,|x'_N|$ are pairwise distinct and such that $\max_i\{|x_i|,|x_i'|\} \leq R$. Then
\begin{gather}
 \frac{1}{2{\rm i}\pi} \oint_{C} \biggl(\frac{\dd z}{z}
{\rm e}^{\sum_{m \geq 0} z^{2m+1}(t_{2m+1}-\tilde{t}_{2m+1})}
\tau\bigl(\mathbf{t}-2\bigl[z^{-1}\bigr];x_1,\dots,x_N\bigr)
\tau\bigl(\tilde{\mathbf{t}}+2\bigl[z^{-1}\bigr];x_1',\dots,x_N'\bigr)\biggr) \nonumber
\\
 =\tau(\mathbf{t};x_1,\dots,x_N)
\tau\bigl(\tilde{\mathbf{t}};x_1',\dots,x_N'\bigr) \nonumber
\\
 -4 \sum_{p=1}^N (-1)^p
\langle 0|\gamma(\mathbf{t}) \phi(x_1) \cdots \widehat{\phi\bigl(x_p\bigr)} \cdots
\phi(x_N) \phi_0|0\rangle
\langle 0|\gamma\bigl(\tilde{\mathbf{t}}\bigr) \phi\bigl(x_p\bigr) \phi(x_1')\cdots \phi(x_N')
\phi_0|0\rangle \nonumber
\\
 -4 \sum_{q=1}^N (-1)^q
\langle 0|\gamma(\mathbf{t}) \phi\bigl(x_q'\bigr)\phi(x_1)\cdots
\phi(x_N) \phi_0|0\rangle
\langle 0|\gamma\bigl(\tilde{\mathbf{t}}\bigr) \phi(x_1')
\cdots \widehat{\phi\bigl(x_q'\bigr)} \cdots \phi(x_N')
\phi_0|0\rangle.\!\label{pre-Hirota2}
\end{gather}
The expectation values are understood as radially
ordered, see \eqref{radialordering}, so that they represent convergent
power series in ratios $\frac{x_i}{x_j}$ when $|x_i|<|x_j|$
\big(and similar ratios involving $x_i'$\big). The Laurent series in $z$
on the left-hand side converges for $R < |z| <\tilde{R}$.
\begin{proof}
 In the right-hand side of \eqref{pre-Hirotabis}
 we anti-commute both $\phi_{\pm j}$ to the right, but there is
 a~distinguished order to proceed. If at some step $\phi_j$ sits left of
$\phi\bigl(x_{\pi(p)}\bigr)$ and
$\phi_{-j}$ left of $\phi\bigl(x'_{\pi'(q)}\bigr)$,
\begin{itemize}\itemsep=0pt
\item we anti-commute $\phi_j$ to the right through
$\phi\bigl(x_{\pi(p)}\bigr)$ if
$\bigl|x_{\pi(p)}\bigr|>\bigl|x'_{\pi'(q)}\bigr|$ or $q-1=N$;
\item we anti-commute $\phi_{-j}$ to the right
 through $\phi\bigl(x'_{\pi'(q)}\bigr)$ if
$\bigl|x'_{\pi'(q)}\bigr|>\bigl|x_{\pi(p)}\bigr|$ or $p-1=N$.
\end{itemize}
The procedure stops at $p{-}1=q{-}1=N$ and
produces after the final step
\[
\sum_{j\in \mathbb{Z}} (-1)^j\langle 0|\Phi \phi_j\phi_0|0\rangle
\langle 0|\Phi' \phi_{-j}\phi_0|0\rangle =\frac{1}{4}
\langle 0|\Phi|0\rangle\langle 0|\Phi'|0\rangle.
\]
In the anti-commutators the sum over $j$ is restored in the other factor,
where it produces
$\phi\bigl(x_{\pi'(q-1)}\bigr) \phi\bigl(x_{\pi(p)}\bigr)\phi\bigl(x_{\pi'(q)}\bigr)$
or $\phi\bigl(x_{\pi(p-1)}\bigr) \phi\bigl(x'_{\pi'(q)}\bigr)\phi\bigl(x_{\pi(p)}\bigr)$, respectively.
The resulting expectation values evaluate by
Wick's theorem \eqref{Wick} to polynomials in the pair
expectations. Under the new condition
that all $|x_i|$, $|x_i'|$ are pairwise different, every
pair expectation becomes a~convergent power series.
This is the assertion \eqref{pre-Hirota2}
with radial ordering \eqref{radialordering} made
explicit.
\end{proof}
\end{Lemma}

\subsection{Integration away from the (anti)diagonal}

Let us introduce the complement of the fat (anti)diagonal, first in a bounded version
\[
D^N_{R,\epsilon}:=\bigl\{\boldsymbol{x} \in \mathbb{R}^{N} \mid
\max_i |x_i| \leq R \ \text{and} \ \min_{i < j} ||x_i|-|x_j||\geq \epsilon\bigr\},
\]
and $\mathbf{1}^{N}_{R,\epsilon}$ be its indicator function. We will soon pass to the unbounded version $\mathbf{1}^N_{\infty,\epsilon}
=\lim_{R\to \infty} \mathbf{1}^N_{R,\epsilon}$.

\begin{Lemma}
\label{thm:BKPlem}
Let $N$ be an even natural number and $\epsilon >0$. Let
$h_1,\dots,h_N$ be continuously differentiable functions on $\mathbb{R}$
such that $h_i$ and $h_i'$ are bounded by a
polynomial, and $\rho$ is a~positive even function admitting moments
of all order. Then the integral $($with expectation value in $\tau$
understood as radially ordered, see \eqref{radialordering}$)$ is a
well-defined formal power series in
$\bigl(t_{2m+1},\tilde{t}_{2m + 1}\bigr)_{m \geq 0}$ which satisfies
\begin{gather*}
 \int_{\mathbb{R}^{2N}} \Res_{z = 0} \frac{\dd z}{z}
{\rm e}^{\sum_{m \geq 0} z^{2m+1}(t_{2m+1}-\tilde{t}_{2m+1})} \tau\bigl(\mathbf{t} - 2\bigl[z^{-1}\bigr];\boldsymbol{x}\bigr) \tau\bigl(\tilde{\mathbf{t}} + 2\bigl[z^{-1}\bigr];\boldsymbol{x}'\bigr) \\
 \qquad\quad{} \times \mathbf{1}_{\infty,\epsilon}^{2N}(\boldsymbol{x},\boldsymbol{x}') \prod_{i = 1}^{N} h_i(x_i)h_i(x_i')\rho(x_i)\rho(x_i') \, \dd x_i \dd x_i' \\
\qquad {}= \int_{\mathbb{R}^{2N}} \tau(\mathbf{t};\boldsymbol{x}) \tau\bigl(\tilde{\mathbf{t}};\boldsymbol{x}'\bigr)\mathbf{1}_{\infty,\epsilon}^{2N}(\boldsymbol{x},\boldsymbol{x}') \prod_{i = 1}^{N} h_i(x_i)h_i(x_i') \rho(x_i)\rho(x_i')\,\dd x_i \dd x_i'.
\end{gather*}
\end{Lemma}
\begin{proof}
Now we consider $\mathbf{t}$, $\tilde{\mathbf{t}}$ formal. Consider a term
\[
\langle 0|\gamma(\mathbf{t}) \phi(x_q')\phi(x_1)\cdots
\phi(x_N) \phi_0|0\rangle
\langle 0|\gamma(\tilde{\mathbf{t}}) \phi(x_1')\cdots \widehat{\phi(x_q')} \cdots
\phi(x_N')
\phi_0|0\rangle
\]
appearing in the last line of
\eqref{pre-Hirota2}, which by Wick's theorem \eqref{Wick}
is anti-symmetric when exchanging
$x_q'\leftrightarrow x_q$. This term is
multiplied by a function symmetric in
$x_q'\leftrightarrow x_q$ and integrated over a
symmetric domain, thus integrating to zero for every $1\leq q\leq N$.
Similarly for the next-to-last line of \eqref{pre-Hirota2}.

In particular, this is valid for integration over the symmetric domain
$D_{R,\epsilon}$ for fixed $R$, $\epsilon$. We now change the perspective
and view \eqref{pre-Hirota2} as formal power series in the variables
$\mathbf{t}$, $\tilde{\mathbf{t}}$. After commuting $\gamma(\mathbf{t})$
and $\gamma(\tilde{\mathbf{t}})$ via \eqref{gammatphiz} to the right,
the extraction of the coefficient of some monomial in
$\mathbf{t}$, $\tilde{\mathbf{t}}$ yields a polynomial of bounded degree
in $x_p$, $x'_q$ and $z$. Next, commuting in the first line of
\eqref{pre-Hirota2} the $\gamma\bigl(\pm 2\bigl[z^{-1}\bigr]\bigr)$ via \eqref{gammatphiz}
to the right we see that only finitely many terms contribute to the
contour integral, and the coefficients are again polynomials in $x_p$
and $x'_q$ of bounded degree. Finally remains the expectation values
which by Wick's theorem \eqref{Wick} factor into polynomials of pair
expectation values $\langle 0|\phi(x_i)\phi(x_j)|0\rangle$. These can
be estimated by a geometric series which is bounded\footnote{Such a
 bound is necessary to apply the dominated convergence theorem. It
 forces us to keep $\epsilon$ until the very end and it will lead to
 Cauchy principal values.} by $1+\epsilon^{-1}\max(|x_i|,|x_j|)$.
Since $\rho$ has finite moments on $\mathbb{R}$, the limit
$R\to \infty$ exists by the dominated convergence theorem. In
particular, the vanishing of the last two lines of \eqref{pre-Hirota2}
after integration remains true for $R\to \infty$.
\end{proof}

\subsection{Regularisation and Pfaffian expression}

\begin{Lemma}
\label{lemfinal} In the same setting as Lemma~$\ref{thm:BKPlem}$, we have
\begin{equation*}
\begin{split}
\Res_{z = 0} \frac{\dd z}{z} {\rm e}^{\sum_{m \geq 0} z^{2m + 1} (t_{2m + 1} - \tilde{t}_{2m + 1})} \mathcal{T}_N\bigl(\mathbf{t} - 2\bigl[z^{-1}\bigr]\bigr) \mathcal{T}_N\bigl(\tilde{\mathbf{t}} + 2\bigl[z^{-1}\bigr]\bigr) = \mathcal{T}_N(\mathbf{t})\mathcal{T}_N\bigl(\tilde{\mathbf{t}}\bigr),
\end{split}
\end{equation*}
where
\begin{gather}
\mathcal{T}_N(\mathbf{t}) = \operatorname{Pf}_{1 \leq i,j \leq N} \biggl( \fint_{\mathbb{R}^2} H_{i,j}(\mathbf{t};x,y) \rho(x)\rho(y)\, \dd x \dd y \biggr), \nonumber\\
H_{i,j}(\mathbf{t};x,y) = {\rm e}^{\sum_{m \geq 0} t_{2m + 1}(x^{2m + 1} + y^{2m + 1})} \cdot \frac{1}{2} \frac{x - y}{x + y}\cdot h_i(x)h_j(y) . \label{Hijt}
\end{gather}
\end{Lemma}
\begin{proof}
Commuting $\gamma(\mathbf{t})$ with \eqref{gammatphiz} to the right and using
Wick's theorem \eqref{Wick}, one can rewrite for even~$N$
\[
\tau(\mathbf{t};x_1,\dots,x_N) = \operatorname{Pf}_{1 \leq i,j \leq N} \biggl({\rm e}^{\sum_{m \geq 0} t_{2m + 1}(x_i^{2m + 1} + x_j^{2m + 1})}\cdot \frac{1}{2} \frac{x_i - x_j}{x_i + x_j}\biggr).
\]
Inserting this in Lemma~\ref{thm:BKPlem} and expanding the two Pfaffians, we get
\begin{gather}
 \int_{\mathbb{R}^{2N}} \Biggl[ \mathbf{1}_{\infty,\epsilon}^{2N}(\boldsymbol{x},\boldsymbol{x}') \prod_{i = 1}^{N} \rho(x_i)\rho(x_i')\dd x_i \dd x_i' \Biggr] \nonumber
 \\
 \qquad\quad{} \times \oint_C\frac{\dd z}{2{\rm i}\pi z} {\rm e}^{\sum_{m \geq 0} z^{2m+1}(t_{2m+1}-\tilde{t}_{2m+1})}
 \sum_{\pi,\pi'\in \mathfrak{S}_N}
 \frac{ (-1)^{\mathrm{sign}(\pi)+\mathrm{sign}(\pi')}}{
 2^{N}(N!)^2} \nonumber
\\
\qquad \quad{} \times \prod_{i=1}^N H_{\pi(2i-1),\pi(2i)}\bigl(\mathbf{t}
-2\bigl[z^{-1}\bigr];x_{\pi(2i-1)},x_{\pi(2i)}\bigr) \nonumber
\\ \qquad \qquad{} \hphantom{\times \prod_{i=1}^N}{}
\times H_{\pi'(2i-1),\pi'(2i)}\bigl(\tilde{\mathbf{t}}+2\bigl[z^{-1}\bigr];
x'_{\pi'(2i-1)},x'_{\pi'(2i)}\bigr) \nonumber
\\
\qquad{} =
\int_{\mathbb{R}^{2N}}
 \Biggl[ \mathbf{1}_{\infty,\epsilon}^{2N}(\boldsymbol{x},\boldsymbol{x}') \prod_{i = 1}^{N} \rho(x_i)\rho(x_i')\dd x_i \dd x_i' \Biggr]
 \sum_{\pi,\pi'\in \mathfrak{S}_N}
 \frac{ (-1)^{\mathrm{sign}(\pi)+\mathrm{sign}(\pi')}}{
 2^{N}(N!)^2} \nonumber
\\
\qquad\quad{} \times \prod_{i=1}^N H_{\pi(2i-1),\pi(2i)}\bigl(\mathbf{t};x_{\pi(2i-1)},x_{\pi(2i)}\bigr)
 H_{\pi'(2i-1),\pi'(2i)}\bigl(\tilde{\mathbf{t}};x'_{\pi'(2i-1)},x'_{\pi'(2i)}\bigr),\label{expanH}
\end{gather}
where $H_{i,j}(\mathbf{t};x_i,x_j)$ was given in (\ref{Hijt}).

We would like to show that both sides have a limit as $\epsilon \rightarrow 0$. To proceed, we decompose in even and odd parts
\[
{\rm e}^{\sum_{m \geq 0} t_{2m + 1}x^{2m + 1}} h_i(x) = h_i^+(x) + xh_i^-(x) \qquad \text{with}\ h_i^{\pm} \ \text{even}.
\]
Then, we restrict the integration over $x_1,\dots,x_N ,x_1',\dots,x_N'\geq 0$ by taking the even part of the integrand in each variable and multiplying by $4$, namely replacing $H_{i,j}(\mathbf{t};x,y)$ with
\[
H_{i,j}^+(\mathbf{t};x,y) = \frac{\bigl(x^2+y^2\bigr)h_i^+(x)h_j^+(y)}{x^2 - y^2} - \frac{2x^2y^2h_i^-(x)h_j^-(y)}{x^2 - y^2}.
\]
Since the domain of integration $D_{\infty,\epsilon}$ is symmetric under $x_i \leftrightarrow x_j$ and $x_i' \leftrightarrow x_j'$, we can also replace each factor of $H_{i,j}(\mathbf{t};x,y)$ in the integral with
\begin{align}
\tilde{H}_{i,j}(\mathbf{t};x,y) ={}& \frac{1}{2}\bigl(H_{i,j}^+(\mathbf{t};x,y) + H_{i,j}^+(\mathbf{t};y,x)\bigr) \label{fromHtoHtilde}\\
={}& \frac{x^2 + y^2}{2(x+y)} \frac{h_i^+(x)h_j^+(y) - h_i^+(y)h_j^+(x)}{x - y} - \frac{x^2y^2}{x+y} \frac{h_i^-(x)h_j^-(y) - h_i^-(y)h_j^-(y)}{x - y}.
\nonumber
\end{align}
If $f$, $g$ are continuously differentiable functions, we can rewrite
\begin{align}
\frac{f(x)g(y) - f(y)g(x)}{x - y} & = \frac{f(x) - f(y)}{x - y} g(y) - f(y) \frac{g(x) - g(y)}{x - y} \nonumber\\
& = \int_{0}^{1}\dd t (g(y)f'((1-t)x + ty) - f(y)g'((1-t)x + ty)).\label{2ratio}
\end{align}
The two products of $H$'s in \eqref{expanH} are therefore replaced by
two products of $\tilde{H}$'s, and we expand them with
\eqref{fromHtoHtilde}--\eqref{2ratio} to get $8^N$ terms integrated
over $\mathbb{R}_{> 0}^{2N} \cap D_{\infty,\epsilon}^{2N}$. For
integration of each of the term, we can use dominated convergence to
let $\epsilon \rightarrow 0$ in both sides of \eqref{expanH} because
\smash{$0<\frac{x^2 + y^2}{x + y}\leq x+y$} and \smash{$0\leq \frac{x^2y^2}{x + y}\leq \frac{1}{16}(x+y)^3$} on
$\mathbb{R}_{\geq 0}$, $h_i^{\pm}$ as well as its first order derivative are
bounded by a polynomial, and $\rho$ has finite moments of all orders.

Now let us compare to the Pfaffian of the integrated kernel
\begin{align*}
\operatorname{Pf} \biggl( \fint_{\mathbb{R}^2} H_{i,j}(\mathbf{t};x,y) \rho(x)\rho(y) \dd x \dd y\biggr) & = \operatorname{Pf} \biggl(\lim_{\epsilon \rightarrow 0} \int_{\mathbb{R}^2}H_{i,j}(\mathbf{t};x,y) \mathbf{1}_{\infty,\epsilon}^{2}(x,y) \rho(x) \rho(y)\, \dd x \dd y \biggr) \\
& = \operatorname{Pf} \biggl(\lim_{\epsilon \rightarrow 0} \int_{\mathbb{R}_{> 0}^2} \tilde{H}_{i,j}(\mathbf{t};x,y) \mathbf{1}_{\infty,\epsilon} ^2(x,y)\rho(x)\rho(y)\, \dd x \dd y\biggr).
\end{align*}
Expanding the Pfaffian, the only difference with the right-hand side of \eqref{expanH} as we handled it is that the $(\pi,\pi')$-term is now integrated against of the indicator
\[
\prod_{i = 1}^{N} \mathbf{1}^2_{\infty,\epsilon}\bigl(x_{\pi(2i - 1)},x_{\pi(2i)}\bigr) \mathbf{1}^2_{\infty,\epsilon}\bigl(x_{\pi'(2i-1)}',x_{\pi'(2i)}'\bigr)
\]
instead of $\mathbf{1}^{2N}(\boldsymbol{x},\boldsymbol{y})$. As the difference between the two indicators converge pointwise to $0$ as $\epsilon \rightarrow 0$, the domination argument used previously to handle each term shows that the limit $\epsilon \rightarrow 0$ exists and is equal to the $\epsilon \rightarrow 0$ limit of the right-hand side of \eqref{expanH}. The same argument applies for the left-hand side, and this yields the claimed identity.
\end{proof}

Finally, Corollary~\ref{Comain} follows from Lemma~\ref{lemfinal}
by setting
$h_j(x) = \sqrt{2} {\rm e}^{-(\lambda_j-\lambda_{\min})\frac{x^2}{2}}$ and
$\rho(x) = {\rm e}^{-\frac{1}{2}\lambda_{\min} x^2 + V_0(x)}$
and observing then
\[
Z_N(\mathbf{t}) = \frac{\sqrt{\Delta(\boldsymbol{\lambda},\boldsymbol{\lambda})}}{2^{\frac{N^2}{2}} (2\pi)^{\frac{N}{2}} \prod_{n = 1}^{N - 1} n!} \mathcal{T}_N(\mathbf{t})
\]
 by comparison with Theorem~\ref{th:1}.

\begin{Remark} \label{RemarkBKP}

All the manipulations until (and including) Lemma~\ref{lemmaRes} are standard. If we were ignoring the regularisation issues,
the BKP tau-function $\mathcal{T}_N(\mathbf{t})$ of \eqref{lemfinal} would
be associated with the group element
\begin{equation}
\label{gGamma}
\hat{g} = \int_{\Gamma^N} \prod_{i = 1}^{N} \phi(x_i) h_i(x_i)\rho(x_i)\,\dd x_i,
\qquad \Gamma = \mathbb{R}
\end{equation}
in the notation of \cite[Section 7.3]{harnad_balogh_2021}. This element
does not naively satisfy the quadratic algebraic condition
\cite[equation 7.3.63]{harnad_balogh_2021}:
there are correction terms corresponding to the last two lines of~\eqref{pre-Hirota2}. There is no contradiction
with \cite[Theorem 7.3.10]{harnad_balogh_2021} (which can already be
found in \cite{Date:1981qz}) as the latter characterises the
polynomial solutions of the BKP hierarchy, while our partition
function is not polynomial. Non-polynomial BKP tau-functions
associated with \eqref{gGamma} were discussed in
\cite{harnad_balogh_2021}, but only for integration contours such that
$\Gamma \cap (-\Gamma) = \varnothing$, and this does not apply to us. If
one wanted to make sense of the quadratic algebraic condition on $\hat{g}$
to formulate a generalisation of \cite[Theorem
7.3.10]{harnad_balogh_2021}, one would need to describe carefully
which completion of the Clifford algebra and which notion of tensor
product should be used. For instance, in our case, we see from the
proof that the quadratic algebraic condition only holds in a
distributional sense (after integration).
\end{Remark}

\section{Comments}\label{section4}

\subsection{Kontsevich cubic model} The original matrix model of Kontsevich \cite{Kontsevich:1992ti} is obtained from \eqref{Zdef} by specialisation to potential $V_0(x) = \frac{{\rm i}x^3}{6}$ and $\mathbf{t} = 0$, and Kontsevich showed that $Z_N(\mathbf{0})$ is a KdV tau-function with respect to the times $\mathbf{s} = (s_1,s_3,s_5,\dots)$ with $s_{2k + 1} = - \frac{1}{2k + 1} \operatorname{Tr}\Lambda^{-(2k +1)}$. By a result of Alexandrov
\cite{Alexandrov:2020nzt}, it is also a BKP tau-function in the times $(2s_1,2s_3,2s_5,\dots)$.

\subsection[The Q-Schur approach and 2-BKP conjecture]{The $\boldsymbol{Q}$-Schur approach and $\boldsymbol{2}$-BKP conjecture}

The BKP hierarchy of Corollary~\ref{Comain} is independent of the KdV/BKP
structure of the original Kontsevich model
since it rather governs the evolution under polynomial deformations of
the potential (parameters $\mathbf{t}$), for fixed $\Lambda$ (parameters $\mathbf{s}$). We are
mainly interested in expansions around $\mathbf{t} =\boldsymbol{0}$, where the BKP hierarchy amounts to quadratic relations
 between moments. When~$V_0$ is even
these relations simplify considerably as we have indicated at the end of Section~\ref{Chap1}. We expect that for arbitrary $V_0$, $Z_N(\mathbf{t})$ is a 2-BKP tau-function in the times $\mathbf{t}$ and $\mathbf{s}$ defined by $s_k = \frac{2}{2k + 1} \mathrm{Tr} \Lambda^{-(2k + 1)}$. Following the suggestion of an anonymous referee, we prove this for ${V_0 = 0}$, using properties of Schur $Q$-functions -- for the specific case of the original Kontsevich matrix model, see also \cite{LY}. In particular, this retrieves in this special case the result of Corollary~\ref{Comain}.

\begin{Theorem}
\label{2BKP} For $V_0 = 0$, $Z_N(\mathbf{t})$ as a function of $\mathbf{t}$ and $\mathbf{s} = \frac{1}{2k + 1}{\rm Tr}\Lambda^{-(2k + 1)}$ is a $2$-BKP tau-function.
\end{Theorem}

\begin{proof}
The starting point is the Cauchy formula for Schur $Q$-functions
\[
 \exp\biggl(2\sum_{k \geq 0} (2k+1) t_{2k+1}p_{2k+1}\biggr)
 =\sum_{\lambda \in \text{SP}} \frac{1}{2^{\ell(\lambda)}}
 Q_\lambda(\mathbf{t})Q_\lambda(\mathbf{p}),
\]
where the sum is over the set $\text{SP}$ of strict integer partitions
$\lambda_1>\lambda_2>\dots > \lambda_{\ell(\lambda)}>0$,
including $Q_\varnothing(\mathbf{t})=1$. Note that there are different conventions for the definition of $Q$-Schur functions in the literature.
The key step are two formulae, that can be found, e.g., in \cite[equations~(55) and~(56)]{Mironov:2020tjf}. Adapted in our notations, the first formula allows computing Gaussian averages of $Q$-Schur functions:
\begin{equation*}
 \int_{\mathcal{H}_N} \dd \mathbb{P}_N(H)
 Q_{2\lambda}\bigl(\bigl\{\tfrac{1}{2k+1}\mathrm{Tr}\bigl(H^{2k+1}\bigr)\bigr\}_{k \geq 1}\bigr)=
 \frac{Q_\lambda(\mathbf{s})
 Q_\lambda(1,0,0,\dots)}{Q_{2\lambda}(1,0,0,\dots)} ,
\end{equation*}
where $2\lambda=(2\lambda_1,2\lambda_2,\dots,2\lambda_{\ell(\lambda)})$
and $\mathbf{s}=(s_1,s_3,s_5,\dots)$
with $s_{2k+1}=\frac{1}{2k+1}\operatorname{Tr}\Lambda^{-(2k+1)}$.
The corresponding integral over $Q_\lambda$ vanishes if $\lambda$
has any odd part. The second formula comes from the analog of the hook-length formula for the specialisation to $(1,0,0,\dots)$ of the $Q$-Schur functions
\[
\frac{Q_\lambda(1,0,0,\dots)}{Q_{2\lambda}(1,0,0,\dots)}
=\prod_{j=1}^{\ell(\lambda)} (2\lambda_j-1)!!.
\]

Combining both formulae yields for $V_0 = 0$
\begin{align}
Z_N(\mathbf{t}) & = \int_{\mathcal{H}_N} \dd \mathbb{P}_N(H)
 \exp\biggl(\sum_{k \geq 0} t_{2k+1} \mathrm{Tr}\bigl(H^{2k+1}\bigr)\biggr) \nonumber\\
& =\sum_{\lambda \in \text{SP}} \frac{1}{2^{\ell(\lambda)}}
Q_{2\lambda}\bigl(\tfrac{1}{2}\mathbf{t}\bigr)
\frac{Q_\lambda(\mathbf{s})
 Q_\lambda(1,0,0,\dots)}{Q_{2\lambda}(1,0,0,\dots)} .\label{ZNt}
 \end{align}

By comparison with \cite[Lemma 5.5]{MMNO}, we recognise a 2-BKP tau-function with
\[
C_{j,k} = \delta_{j,2k} (2k - 1)!!
\]
in the times $\mathbf{t}$ and $2\mathbf{s}$. The factor of $2$ in the times is necessary to compare with the BKP hierarchy with the factors specified in \eqref{biline} for the BKP equations, while it was absent in \cite{MMNO} due to the normalisation chosen in their equation~(3.22).
\end{proof}

\subsection{Cartan--Pl\"ucker relations} It is known (see, e.g., \cite[Theorem~7.1.1]{harnad_balogh_2021}) that a formal power series in $\mathbf{t}$ is a
BKP tau-function if and only if it can be expanded as
\[
\sum_{\lambda \in SP} \kappa_\lambda
Q_\lambda\bigl(\tfrac{1}{2}\mathbf{t}\bigr),
\]
where the $\kappa_\lambda$ satisfy Cartan--Pl\"ucker relations for isotropic Grassmannians. Stopping at \eqref{ZNt}, we see that Corollary~\ref{Comain} is equivalent to the property that, \emph{for any potential
$V_0(x)$ for which ${\rm e}^{- \frac{1}{2}\lambda_{\min} x^2 + V_0(x)}$ has
finite moments}, the family
\begin{equation*}
\kappa_\lambda= \int_{\mathcal{H}_N} \dd \mathbb{P}_N(H)
 Q_{\lambda}\bigl(\bigl\{\tfrac{1}{2k+1}\mathrm{Tr}\bigl(H^{2k+1}\bigr)\bigr\}\bigr) {\rm e}^{\mathrm{Tr}(V_0(H))}
\end{equation*}
satisfies the Cartan--Pl\"ucker relations for isotropic Grassmannians. For general $V_0$, this result
does not seem to be covered by previous work. For $V_0 = 0$, it is covered alternatively by Theorem~\ref{2BKP} and corresponds to
\[
\kappa_\lambda= Q_{\frac{\lambda}{2}}(\mathbf{s}) \prod_{i = 1}^{\ell(\lambda)} (2\lambda_i - 1)!!.
\]

\subsection{Topological recursion} Apart from $V_0 = 0$ and $V_0$ cubic, the simplest even case is
$V_0(x)= -\frac{cN}{4}x^4$ for some parameter $c > 0$, and its formal
large $N$ topological expansion has been studied during the last years~\mbox{\cite{Grosse:2019jnv, Schurmann:2019mzu}},
providing strong evidence
\cite{Branahl:2020yru} that the topological expansion of the cumulants
obey the blobbed topological recursion \cite{Borot:2015hna}, which is
the general solution of abstract loop equations \cite{Borot:2013lpa}.
In~\cite{Hock:2023nki}, a~recursive formula for meromorphic
differentials which are generating series of the genus~$1$ cumulants
was given, and a generalisation to higher genera was outlined.

On the other hand, BKP tau-functions of hypergeometric type with mild analytic assumptions are known to satisfy abstract loop equations, and thus (perhaps blobbed) topological recursion~\cite{SA}. In particular, this was applied to prove the conjecture of \cite{GKL} that spin Hurwitz numbers (weighted by the parity of a spin structure) satisfy topological recursion. Although~$Z_N(\mathbf{t})$ is not a hypergeometric tau-function of BKP in the sense of \cite{Orlov}, one may ask if similar techniques could not prove that the topological expansion of the correlators of $Z_N(\mathbf{t})$ are governed by (blobbed or not) topological recursion. We hope to return to this question at a later occasion.\looseness=1

\subsection*{Acknowledgements}

R.W.\ was supported by the Cluster of Excellence
\emph{Mathematics M\"unster} and the CRC 1442 \emph{Geometry:
 Deformations and Rigidity} (Funded by the Deutsche
 Forschungsgemeinschaft
 Project-ID 427320536 -- SFB 1442, as well as under Germany's
 Excellence Strategy EXC 2044 390685587, Mathematics M\"unster:
 Dynamics -- Geometry -- Structure). Parts of this note were prepared during
a workshop at the Erwin Schr\"odinger Institute in Vienna, and during a stay at IH\'ES of G.B: we thank these institutions for their hospitality. We thank anonymous referees for valuable comments which led to expand Sections~\ref{SproofBKP} and~\ref{section4}.

\newpage

\pdfbookmark[1]{References}{ref}
\LastPageEnding

\end{document}